\documentclass[11pt,letterpaper]{article}
\usepackage{fullpage,xspace,times}
\usepackage{graphicx}
\usepackage{amsmath,upgreek}
\usepackage{amsfonts}
\usepackage{amssymb}
\usepackage{epstopdf}

\newtheorem{theorem}{Theorem}

\newtheorem{lemma}[theorem]{Lemma}

\newcommand{\Xomit}[1]{ }
\newenvironment{proof}[1][Proof]{\textbf{#1.} }{\ \rule{0.5em}{0.5em}}

\mathchardef\mhyphen="2D

\newcommand{\eps}{\upvarepsilon}

\def\eps{\varepsilon}

\begin{document}




\title{Several methods of analysis \\ for cardinality constrained bin packing}

\date{}

\author{Leah Epstein\thanks{
Department of Mathematics, University of Haifa, Haifa, Israel.
\texttt{lea@math.haifa.ac.il}. }}

\maketitle



\begin{abstract}
We consider a known variant of bin packing called {\it cardinality constrained bin packing}, also called {\it bin packing with cardinality constraints} (BPCC). In this problem, there is a parameter $k\geq 2$, and items of rational sizes in $[0,1]$ are to be packed into bins, such that no bin has more than $k$ items or total size larger than $1$. The goal is to minimize the number of bins.

A recently introduced concept, called the price of clustering, deals with inputs that are presented in a way that they are split into clusters. Thus, an item has two attributes which are its size and its cluster.
The goal is to measure the relation between an optimal solution that cannot combine items of different clusters into bins, and an optimal solution that can combine items of different clusters arbitrarily. Usually the number of clusters may be large, while clusters are relatively small, though not trivially small.
Such problems are related to greedy bin packing algorithms, and to batched bin packing, which is similar to the price of clustering, but there is a constant number of large clusters.
We analyze the price of clustering for BPCC, including the parametric case with bounded item sizes. We discuss several greedy algorithms for this problem that were not studied in the past, and comment on batched bin packing.
\end{abstract}

\section{Introduction}
The bin packing problem is defined as follows. Given items of rational sizes in $[0,1]$, partition these items to subsets of total sizes not exceeding $1$ so as to minimize the number of subsets. The subsets are called bins, and the process of assigning an item to a subset is called packing. In the classic or standard variant, there are no other attributes or conditions.
This problem has been studied for fifty years \cite{J74,JoDUGG74,FerLue81,BC81,KK82,LeeLee85,RaBrLL89,Vliet92,BBG,DS12,DLHT,Dosa15,BBDEL_ESA18,BBDEL_newlb,BBDSS19}.
Another well-known natural variant limits the number of items that one bin can receive. Given an integer parameter $k\geq 2$, bin packing with cardinality constraints (BPCC) is the variant of bin packing where every bin is required to have at most $k$ items \cite{KSS75,KP99,CKP03,BCKK04,Epstein05,EL07afptas,FK13,BDE,DosaE18,BBDEL_ESA}. In this work we study BPCC with respect to new and old concepts.
One well-studied concept is greedy algorithms. Such algorithms are frequently online, in the sense that they pack every item before seeing the future items. The problems were studied as offline problems, where an algorithm  receives the entire input as a set, and as online problems, where an algorithm receives the input as a sequence. Here, we study not only algorithms but other concepts (defined below) as well.

\vspace{0.1cm}

\noindent{\bf Measures for bin packing.}
For an algorithm for a certain problem, that does not necessarily compute an optimal solution, we say that its asymptotic approximation ratio does not exceed $R$, if there
exists a constant value $C\geq 0$ (this value has to be independent of the input),
such that for any input $J$ for the problem, the cost of this algorithm for input $J$
is bounded from above by the following value: $R$ times the optimal cost
for $J$ plus  the constant $C$. The asymptotic approximation ratio is the infimum value $R$ for which this inequality holds for every input $J$.
If the constant $C$ is equal to zero, the approximation ratio will be called strict or absolute.
One specific optimal (offline) algorithm for the problem and its cost are usually denoted by
$OPT$, and $OPT(J)$ is also used for the cost for a fixed input $J$. Another
definition of the asymptotic approximation ratio is the supreme
limit of the ratio between the cost of the algorithm and the cost of $OPT$, as a
function of the second cost, where we take the maximum or supremum over
the inputs with the same optimal cost.
For online algorithms, we still use the term {\it approximation ratio}, and the definition is unchanged. The term {\it competitive ratio} has the same meaning in the literature.
The asymptotic measures are considered to be more meaningful for bin packing problems, and thus we are mostly interested in those asymptotic measures. In the cases where we do not specify whether the measure is asymptotic or absolute, the result holds for both these measures. We will use such an analysis throughout the article.
\vspace{0.1cm}

\noindent{\bf Parametric variants for bin packing problems.}
It is often assumed that the size of an item may be very close to $1$ or even equal to $1$. There are applications where this may happen, but in other applications, bin capacities are much larger than item sizes. The parametric case is defined by an upper bound on item sizes. Specifically, there is a parameter $\beta$,  where $0<\beta\leq 1$, such that items sizes are rational numbers in $[0,\beta]$. For some versions of bin packing and some algorithms for it, the study of the parametric case is different from the general case, and it is of interest. There are models where  the interesting cases are only those where $\frac{1}{\beta}$ is an integer, but there may be additional relevant values of $\beta$.
The resulting asymptotic approximation ratio may or may not be close to $1$ for small values of $\beta$ \cite{JoDUGG74,J74,ACFR16}.
\vspace{0.1cm}

\noindent{\bf Bin packing with item types: Batched bin packing.}
An item type is an additional attribute of an item. This attribute corresponds to a fixed partition of the input into subsets called batches or clusters. The number of batches or clusters will be denoted by $\ell$, where $\ell\geq 1$.
Batched bin packing is a model that was initially defined as a semi-online input arrival scenario for bin packing problems. In this scenario, an algorithm receives all items of one type together. There is work on the version where the algorithm can pack items of different types into one bin (so the difficulty is the lack of knowledge on future batches) \cite{GJYbatch,BBDGT,Do15,Epstein16}, and here we consider the variant where items of different batches are to be packed into separate bins  \cite{Do15,Epstein16}.  Studies focus on the difference between solutions where items of different batches are packed separately (batched solutions), and globally optimal solutions that can pack items of different types into the same bin. We use the asymptotic approximation ratio to compare a batched optimal solution with a globally optimal solution. This is a comparison between optimal solutions, but sometimes it is easier to analyze greedy solutions instead.
The problem was studied for several variants of bin packing \cite{Epstein16,E21}, and it turns out that for $\ell=2$, the tight bound on the asymptotic approximation ratio is $1.5$ \cite{Epstein16}, and for larger values of $\ell$ the ratio increases and grows to approximately $1.69103$ (see \cite{LeeLee85} for an early article where this value appeared in another context of bounded space online bin packing).
\vspace{0.1cm}

\noindent{\bf Bin packing with item types: The price of clustering.}
This problem is similar to batched bin packing, only here $\ell$ may be large while sets of items of one type may be small. A solution where different item types have separate bins is called a clustered solution, and it is compared to a globally optimal solution in this case as well. Thus, a clustered optimal solution is compared to a globally optimal solution using the approximation ratio once again, only here both the absolute approximation ratio and the asymptotic approximation ratio are studied. This approximation ratio is called {\it The price of clustering} (PoC) \cite{AESV,E20}.  Azar et al. \cite{AESV} who introduced this problem, write that (in order for the measure to be meaningful) clusters may be small, but they cannot be arbitrarily small. The assumption that an optimal solution for every cluster has at least two bin leads to tight bounds of $2$ on the PoC \cite{AESV}. This bound is tight even if items can be very small. For the cases where the optimal solution for every cluster requires at least $q$ bins with $q\geq 3$, the approximation ratio is strictly smaller than $2$  \cite{AESV,E20} and decreases to approximately $1.69103$ as $q$ grows \cite{Epstein16,AESV}. The PoC was also studied for other variants of bin packing \cite{E21}. Here, in the study of CCBP, we use the smallest lower bound on $q$, i.e., we let $q=2$. Our assumption for every cluster is that the optimal cost for it is at least $2$. Due to the cardinality constraints, the total size of items can still be not larger than $1$ for some clusters, and for such a cluster it will hold that the number of items is at least $q+1$.
\vspace{0.1cm}

\noindent{\bf Notation for bin packing problems with types.}
Given an input, recall that $\ell$ is the number of clusters or batches. We let
$OPT_i$ be the number of bins in an optimal solution for the $i$th
cluster or batch, and we denote the input (the set of items of type $i$) by $I_i$. We let $I$ be the set of all items, that is,
$I=\bigcup_{1 \leq i \leq \ell} I_i$, where $n=|I|$. We use the notation above for optimal solutions, and let $OPT$ be
a globally optimal solution for $I$, as well as its cost. We let $A_i$ be the number of bins in the output of a fixed algorithm (which can any algorithm, an optimal algorithm or another algorithm) for cluster $i$.
\vspace{0.1cm}

\noindent{\bf Greedy algorithms and related work.}
The study of greedy algorithms for bin packing started together with the first studies of classic bin packing and  \cite{JoDUGG74,J74}.
Next Fit (NF) is an efficient algorithm that has (at most) one active bin where items are packed, and once an item cannot be added to the bin, it is closed and replaced with a new active bin.
Any Fit (AF) algorithms pack an item into a non-empty bin when possible, and otherwise a new bin is used. A specific AF algorithm can be defined based on the choice of bin when there are several options. First Fit (FF) chooses the first (minimum index) bin, while Worst Fit (WF) chooses the least packed bin in terms of total size \cite{J74,JoDUGG74}.
It is known that the asymptotic and absolute approximation ratio of FF is $1.7$ \cite{JoDUGG74,DS12}. For WF (and AF in general) and NF, this ratio is equal to $2$ \cite{J74}.
These algorithms (for classic bin packing) and online algorithms were sometimes analyzed for the parametric case as well  \cite{JoDUGG74,J74,Vliet92,BBG}.
The asymptotic approximation ratio of FF with $\beta \in (\frac 1{t+1},\frac 1t]$ for an integer $t\geq 2$ is $\frac{t+1}t$ \cite{JoDUGG74},
and the case where the parameter $\beta \in (\frac 1{t+1},\frac 1t)$ is not different from the case $\beta = \frac 1t$ (also for $t=1$). For NF and WF, the asymptotic competitive ratio is $2$ for $\beta>\frac 12$. while for $\beta\leq \frac 12$, the ratio is exactly $\frac{1}{1-\beta}$ \cite{J74}.

Greedy algorithms can be applied to CCBP. The algorithm First Fit (FF) for this variant is defined as follows. Every item in the list is packed into the bin of the minimum index such that the bin has sufficient space for it and less than $k$ items. Similarly, we can adapt other greedy algorithms, where the concept of the possibility of packing an item into a bin is tested with respect to the total size (which cannot exceed $1$ together with the new item) and the number of items (which cannot exceed $k$ together with the new item, that is it has to be at most $k-1$ prior to the packing).
CCBP was analyzed for greedy algorithms already in the article where the problem was introduced \cite{KSS75}. In that work, FF and its sorted version FFD were studied. For FF, the asymptotic approximation ratio tends to $2.7$ as $k$ grows. The exact asymptotic approximation ratio for every value of $k$ was found much later, and it consists of several cases according to the value of $k$ \cite{DosaE18}, where, in particular, for $k\geq 10$ the ratio is $2.7-\frac{3}k$.
Note that greedy algorithms do not typically have the best possible asymptotic approximation ratios for online algorithms for bin packing problems \cite{BBDEL_ESA18,BCKK04,KSS75,JoDUGG74,J74}.
\vspace{0.1cm}

\noindent{\bf Structure of the paper and results.}
In this work we study two kinds of problems. The first one is CCBP with types, where we study the price of clustering and batched bin packing. The second one is greedy algorithms. In most cases, the analysis includes the parametric case. Specifically, in Section \ref{typ} we find the exact price of clustering for any $k$ and $\beta$ (and $q=2$), we analyze batched bin packing with two and three batches, and we explain why large numbers of batches are uninteresting. In Section \ref{gree} we analyze NF and WF for the general case, and using sample cases for $\beta$ we exhibit the difficulty of analysis for the parametric case. We also provide a simple analysis of the parametric case for FF, for which the general case was studied in the past \cite{KSS75,DosaE18}.




\section{The price of clustering and batched bin packing}\label{typ}
We start with a complete study of the PoC. Recall that we assume that for every cluster the optimal solution has at least two bins.
For the analysis of upper bounds, we consider the action of greedy algorithms on the clusters, since the output of such an algorithm for each cluster has a simpler structure compared to an optimal solution for it, and the number of bins obviously cannot be smaller.
We will use weight functions in the analysis. Such a function can relate two solutions as follows. We let the weight of a bin be the sum of weights of its items.
If in one solution the total weight of every bin is at most $\rho$ and in another solution the total weight for every bin is at least $1$ (where it is sufficient to prove this on average), then the ratio between the costs of the two solutions (for one input, for which the total weight is fixed) is at most $\rho$. If there is an additive term, that is, the requirement holds possibly excluding a constant number of bins, this results in an upper bound on the asymptotic approximation ratio.

Recall that in the parametric case we assume that item sizes are in $(0,\beta]$ for some $0<\beta\leq 1$. We let $t=\lfloor \frac{1}{\beta}\rfloor$, so all items have sizes of at most $\beta \leq \frac 1{t}$, but $\beta>\frac{1}{t+1}$. We are interested in all integers $t\geq 1$.
As discussed earlier, for classic bin packing, even if all items are very small, the PoC for the case where optimal solutions for clusters have at least two bins is unchanged.
Now, we will see that for CCBP the bounds become slightly smaller as $\beta$ decreases. Specifically, we will find a function of both $t$ and $k$ such that the PoC is equal to this function. The bounds are between approximately $2$ and approximately $4$, for large values of $k$.
As mentioned above, for certain bin packing problems and their analysis as a function of $\beta$ (or $t$), it is sometimes the case that the ratio depends not only on $t$ but it depends on the exact value of $\beta$. \cite{E21}, but here were prove that the dependence is just on $t$, and the ratio is equal for all values $\beta$ with the same value of $t$.
In the next theorem we provide a complete analysis of the PoC for the general case and the parametric case.

\begin{theorem}
The PoC for CCBP is equal to $\frac{4k-2}{k+1}$.
For the parametric case, the PoC is equal to  $4-\frac{2t+4}{k+1}=\frac{4k-2t}{k+1}$ if $t\leq k$, and to $\frac{2k}{k+1}$ if $t\geq k$.
\end{theorem}
\begin{proof}
The the proof of the theorem consists of four lemmas. The first two lemma provide a proof for the general case, and the next two (whose proofs appear in the Appendix) extend the result for the parametric case.
The first lemma is a proof of the upper bound for the general case.
\begin{lemma}\label{ll1}
For any input $I$ for CCBP with the parameter $k \geq 2$, it holds that $\sum_{j=1}^{\ell} OPT_j \leq (4-\frac{6}{k+1})\cdot OPT=\frac{4k-2}{k+1}\cdot OPT$.
\end{lemma}
\begin{proof}
We use the following weight function $w(x):[0,1)\rightarrow (0,2)$, where $$w(x)=\frac{(2k-2)x+2}{k+1}=\frac{2(k-1)x}{k+1}+\frac 2{k+1} \ . $$ The motivation for it is that due to the cardinality constraint, the weight for a bin should be affected both by the total size and the number of items.

For a bin of the optimal solution, the total size is at most $1$ and the number of items is at most $k$, so the total weight is at most $\frac{(2k-2)+2k}{k+1}=\frac{4k-2}{k+1}$.

Consider a cluster for $I_j$, and an solution of FF for it, which contains $A_j\geq 2$ bins. If all bins, possibly excluding the last bin, have $k$ items, the number of items for $I_j$ is at least $k\cdot (A_j-1)+1$, and
total weight is at least $\frac {2}{k+1} \cdot (k\cdot A_j-k+1)= A_j+\frac{k-1}{k+1}\cdot A_j-\frac{2k-2}{k+1}$. By $A_j\geq 2$ we get $\frac{k-1}{k+1}\cdot A_j-\frac{2k-2}{k+1}\geq 0$, and the total weight for the cluster is at least $A_j$.

Otherwise, let $\alpha_j$ (where $1 \leq \alpha \leq A_i-1$) be the number of bins that are not the last bin, such that every such bin has less than $k$ items at termination. For these bins together with the last bin we have that every pair of bins have a total size above $1$ together. This holds because the first item of the bin with the larger index of the two was not packed into the bin with the smaller index, though that bin had less than $k$ items at termination and therefore also at the time of packing for this item. Thus, the total size of items in these $\alpha_j+1$ bins is above $\frac{\alpha_j+1}2$. Every bin has at least one item, so these $\alpha_j$ bins together have at least $\alpha_j$ items, while the other $A_j-\alpha_j-1 \geq 0$ bins have $k$ items each. The number of items is at least $k\cdot A_j-(k-1)(\alpha_j+1)$, and the total weight is at least $$\frac{1}{k+1}(2(k-1)\cdot \frac{\alpha_j+1}2+2\cdot(k\cdot A_j-(k-1)(\alpha_j+1)))$$ $$=\frac{1}{k+1}((k-1)\alpha_j+k-1+2kA_j-2(k-1)\alpha_j-2(k-1))=A_j+\frac{k-1}{k+1}(A_j-\alpha_j-1) \geq A_j \ . $$

In both cases, the average weight for any bin of the clustered solution is at least $1$, and the weight of any bin of a globally optimal solution is at most $\frac{4k-2}{k+1}$, which implies the claim.
\end{proof}

The second lemma is a proof of the lower bound for the general case.
\begin{lemma}\label{ll2}
The PoC of CCBP with the parameter $k \geq 2$ is at least $4-\frac{6}{k+1}=\frac{4k-2}{k+1}$.
\end{lemma}
\begin{proof}
Let $N>k$ be a large positive integer.

The input consists of $N(k+1)(k-2)$ small items whose sizes are equal to $\frac{1}{N^5}$, and for all values of $i$ such that $1\leq i \leq N(k+1)$, there are two items of sizes $\frac 12+\frac{k\cdot i}{N^5}$ and $\frac 12-\frac{k\cdot (i+1)}{N^5}$. An optimal solution for the entire input packs $N(k+1)$ bins, where the $i$th bin has the two items of sizes $\frac 12+\frac{k\cdot i}{N^5}$ and $\frac 12-\frac{k\cdot (i+1)}{N^5}$, whose total size is $1-\frac{k}{N^5}$, and $k-2$ items small items, whose total size is below $\frac{k}{N^5}$.

The number of clusters will be $N(2k-1)-1$. There are $N(k-2)$ clusters, each with $k+1$ small items, $N(k+1)-2$ clusters with a pair of items $\frac 12+\frac{k\cdot i}{N^5}$ and $\frac 12-\frac{k\cdot (i-1)}{N^5}$ for $3\leq i \leq N(k+1)$, and one cluster with four remaining items of sizes $$\frac 12+\frac{k}{N^5} \ , \ \ \ \frac 12+\frac{2k}{N^5}  \ , \ \ \ \frac 12-\frac{k(N(k+1)+1)}{N^5}\ , \mbox{\ \ \ and \ \ \ } \frac 12-\frac{k(N(k+1))}{N^5} \ . $$

For the clusters with small items, the number of items is larger than $k$, and therefore the optimal cost for every cluster is $2$. For the other clusters, the total size of items is at least $1+\frac{k}{N^5}$, and therefore the optimal cost is also above $1$. Thus, the total number of bins is at least  $2(N(2k-1)-1)$. In fact, one can pack every cluster into two bins. The cost of an optimal solution for the entire input is $N(k+1)$, so the PoC is at least $\frac{4Nk-2N-2}{Nk+N}=\frac{4k-2-2/N}{k+1}$, which tends to  $\frac{4k-2}{k+1}$ as $N$ grows to infinity.
\end{proof}

For any $\beta$ such that $t=1$, we already found a tight result of $\frac{4k-2}{k+1}$ on the PoC. The upper bound of Lemma \ref{ll1} obviously holds, and the lower bound construction of Lemma \ref{ll2} holds since we can use $N>\frac{1}{\beta-1/2}$, which means that items of sizes just above $\frac 12$ are chosen such that they are still smaller than $\beta$.
Thus, we consider the case $t\geq 2$ in what follows, both for the upper bound and the lower bound.

\begin{lemma}\label{ex1}
For any input $I$ for CCBP with the parameter $k \geq 2$, and any $t\geq 2$, it holds that $\sum_{j=1}^{\ell} OPT_j \leq (4-\frac{2t+4}{k+1})\cdot OPT=\frac{4k-2t}{k+1}\cdot OPT$ for $k\geq t$ and  $\sum_{j=1}^{\ell} OPT_j \leq \frac{2k}{k+1}\cdot OPT$ if $k\leq t$.
\end{lemma}
\begin{proof}
We will analyze clusters via FF once again, so we can use properties of FF in the analysis.

For $k \leq t$, we use the weight function $w(x):[0,1)\rightarrow (0,1)$, where $w(x)=\frac{2}{k+1}$. In this case, for every cluster every bin except for possibly one bin will have at least $k$ items. In this case the total size of items for every bin may be arbitrarily small.

For $k \geq t$, we use the weight function $w(x):[0,1)\rightarrow (0,2)$, where $w(x)=\frac{(2k-2t)x+2}{k+1}=\frac{2(k-t)x}{k+1}+\frac 2{k+1}$. In this case, once again the motivation is that the weight for a bin should be affected both by the total size and the number of items. However, in the parametric case, for every cluster, every bin except for possibly one bin will have at least $t$ items, since no item is larger than $\frac 1t$, and therefore FF will not use an empty bin while there is a bin with less than $t$ items.

For a bin of the optimal solution, in the first case there are at most $k$ items and the total weight is at most $\frac{2k}{k+1}$. In the second case,
the total size is at most $1$ and the number of items is at most $k$, so the total weight is at most $\frac{(2k-2t)+2k}{k+1}=\frac{4k-2t}{k+1}$.

Consider a cluster for $I_j$, and an solution of FF for it, which contains $A_j\geq 2$ bins. Every bin except for possibly the last one has $\min\{k,t\}$ items, by the definition of FF.
If every bin excluding the last one has $k$ items, then the total weight is at least $\frac{2}{k+1}\cdot (k(A_j-1)+1) \geq A_j$, by $A_j\geq 2$.

In the case $k \leq t$, indeed every bin except for possibly the last one has $k$ items, so the total weight is at least $A_j$.
We consider the case $k \geq t$, in which every bin except for possibly the last one has at least $t$ items.
If all bins, possibly excluding the last bin, have $k$ items, the total weight is at least $A_j$, and we focus on the case where this does not hold.

Let $\alpha_j \geq 1$ be the number of bins that are not the last bin, where less than $k$ items are packed. For these bins together with the last bin we have that every pair of bins have a total size above $1$ together. This holds because the first item of the bin with the larger index of the two was not packed into the bin with the smaller index, though that bin had less than $k$ items at the time of packing of this item. Moreover, since every item has size not exceeding $\frac 1{t}$, each one of the bins that are not the last bin has a total size of items above $1-\frac 1t$. Thus, $\alpha_j-1 \geq 0$ bins have a total size above $1-\frac 1t$ and the last bin together with another bin have a total size above $1$, and the total size of items in these $\alpha_j+1$ bins is above $(\alpha_j-1)\cdot (1-\frac 1t)+1$. All these bins excluding the last one have at least $t$ items each (for the last bin we can only conclude that it has at least one item), while the other $A_j-\alpha_j-1 \geq 0$ bins have $k$ items each. The number of items is at least $k\cdot A_j-(k-t)\alpha_j-(k-1)$, and the total weight is at least $$\frac{1}{k+1}((2(k-t)((\alpha_j-1)(1-1/t)+1))+2\cdot(k\cdot A_j-(k-t)\alpha_j-(k-1))) \ . $$ We will show that the weight is at least $A_j$ by showing that this expression minus $A_j$ is not negative, i.e.,
$$\frac{1}{k+1}((2(k-t)(\alpha_j-1)(1-1/t)+2(k-t)+(k-1)\cdot A_j-2(k-t)\alpha_j-2(k-1))\geq 0 \ . $$
By $A_j \geq \alpha_j-1 $, it is sufficient to show that $$2(k-t)(\alpha_j-1)(1-1/t)+2(k-t)+(k-1)\cdot (\alpha_j+1)-2(k-t)\alpha_j-2(k-1)\geq 0$$ holds, which is equivalent to
$$2(k-t)((\alpha_j-1)(1-1/t)+1-\alpha_j)+(k-1)\cdot (\alpha_j-1)\geq 0 \ , $$ and to $$(\alpha_j-1)(2(k-t)(-1/t)+(k-1))\geq 0 \ . $$ Since $\alpha_j\geq 1$ it is left to prove that $2(k-t)(-1/t)+(k-1) \geq 0$. Since $t\geq 2$, we have
$$2(k-t)(-1/t)+(k-1)=-\frac{2k}t+2+k-1\geq 1 \ . $$\end{proof}

\begin{lemma}\label{ex2}
The PoC of CCBP with the parameter $k \geq 2$ and any $t \geq 2$ is at least $4-\frac{2t+4}{k+1}=\frac{4k-2t}{k+1}$ if $t\leq k$, and at least  $\frac{2k}{k+1}$ if $t\geq k$.
\end{lemma}
\begin{proof}
Let $N>k$ be a large positive integer. We also require that $N>\frac{1}{\beta-1/(t+1)}$, so that all sizes will not exceed $\beta$.

In the case $t\geq k$, there are only items of sizes $\frac 1t$, where $\frac 1t\leq \frac 1k$ and $k \cdot \frac 1t \leq 1$. The number of such items is $k(k+1)\cdot N$. An optimal solution for the entire input has $(k+1)N$ bins with $k$ items packed into each bin. Every cluster will have $k+1$ items, and there are $kN$ clusters. An optimal solution for every cluster has two bins. Thus, the PoC is at least $\frac{2kN}{(k+1)N}=\frac{2k}{k+1}$, as required. Note that the construction gives a PoC of $\frac{2k}{k+1}$ already for $N=1$, but using large values of $k$ provides us with an asymptotic lower bound rather than an absolute one.

In the case $t\leq k-1$, the input consists of $N(k+1)(k-t-1)$ small items whose sizes are equal to $\frac{1}{N^5}$, $N(k+1)(t-1)$ items of size $\frac{1}{t+1}$, and for all values of $i$ such that $1\leq i \leq N(k+1)$, there are two items of sizes $\frac 1{t+1}+\frac{k\cdot i}{N^5}$ and $\frac 1{t+1}-\frac{k\cdot (i+1)}{N^5}$. An optimal solution for the entire input packs $N(k+1)$ bins, where the $i$th bin has the two items of sizes $\frac 1{t+1}+\frac{k\cdot i}{N^5}$ and $\frac 1{t+1}-\frac{k\cdot (i+1)}{N^5}$, whose total size is $\frac{2}{t+1}-\frac{k}{N^5}$, $t-1$ items of sizes $\frac{1}{t+1}$, and $k-t-1$ items small items, whose total size is below $\frac{k}{N^5}$.

The number of clusters will be $N(2k-t)-1$. There are $N(k-t-1)$ clusters, each with $k+1$ small items, $N(k+1)-2$ clusters with $t+1$ items each, out of which $t-1$ have size $\frac{1}{t+1}$, and there is also a pair of items $\frac 1{t+1}+\frac{k\cdot i}{N^5}$ and $\frac 1{t+1}-\frac{k\cdot (i-1)}{N^5}$ for $3\leq i \leq N(k+1)$, and one cluster with four remaining items of sizes $\frac 1{t+1}+\frac{k}{N^5}$, $\frac 1{t+1}+\frac{2k}{N^5}$, $\frac 1{t+1}-\frac{k(N(k+1)+1)}{N^5}$, and $\frac 1{t+1}-\frac{k(N(k+1))}{N^5}$, and the remaining $2t-2$ items of size $\frac{1}{t+1}$.

For the clusters with small items, the number of items is larger than $k$, and therefore the optimal cost for every cluster is $2$. For the other clusters, the total size of items is above $1+\frac{k}{N^5}$, and therefore the optimal cost is also above $1$. Thus, the total number of bins is at least  $2(N(2k-t)-1)$ since one can pack every cluster into two bins. The cost of an optimal solution for the entire input is $N(k+1)$, so the PoC is at least $\frac{4Nk-2Nt-2}{Nk+N}=\frac{4k-2t-2/N}{k+1}$, which tends to  $\frac{4k-2t}{k+1}$ as $N$ grows to infinity.
\end{proof}

Combining the last two lemma concludes the proof of the theorem.
\end{proof}

\vspace{0.1cm}

Next, we consider batched bin packing with a small number of batches.
We will analyze the case with two and three batches, as a function of $k$. For a large number of batches the problem becomes similar to batch bin packing for classic bin packing \cite{Do15,Epstein16}, though the bound is larger by an additive $1$.
Specifically, it is known that for CCBP, for large values of $k$, the ratios may grow by an additive factor of $1$, since small items are sometimes treated almost independently \cite{Epstein05}. For FF and other greedy algorithms, in the bad examples, these items can be presented first, and in the analysis, bins with $k$ items can be considered separately \cite{KSS75,DosaE18}. This property of an additive $1$ in the ratio is not true for online algorithms in general \cite{BCKK04,BDE}, where the tight asymptotic approximation ratio is $2$, (so it is smaller than $1$ plus the best possible asymptotic approximation ratio \cite{Vliet92,BBG,BBDEL_newlb}).
The situation for batched bin packing is similar to the case with an additive $1$.
As in the price of clustering, a bad situation is where all very small items are separated from other parts of the input into a separate batch.
This adds $1$ to the asymptotic approximation ratio for a large number of batches, which becomes approximately $2.69$, since for classic bin packing this value is approximately $1.69$ \cite{LeeLee85,Epstein16}.

The proof of the next theorem appears in the Appendix.
\begin{theorem}\label{ex3}
The asymptotic approximation ratio for batched bin packing and CCBP is $2-\frac 1k$ for two batches and $2.5-\frac 2k$ for three batches.
\end{theorem}
\begin{proof}
The lower bounds are simple. Let $N>0$ be a large integer that is divisible by $2k$. For two batches, consider a solution with $N$ bins, where every bin has one large item of size $0.55$ and $k-1$ small items of size zero.
Splitting this input into two batches with identical items each, the cost of the batch of large items is $N$, and the second batch has bins with $k$ small items, so its cost is $\frac{(k-1)N}{k}$. The ratio is $\frac{2N-N/k}{N}=2-\frac 1k$.
For the case of three batches, add medium items of size $0.4$ to the bins of the optimal solution, each replacing one small item. There is a third batch with medium items, and it contains $\frac N2$ bins. The ratio becomes
$\frac{N+(k-2)N/k+N/2}{N}=2.5-\frac 2k$.

For the upper bound, consider several special cases first.
We show that for $k=2$, the tight bound for any constant number of batches $q \geq 2$ is $1.5$. Consider the weight function $w(x)=1$ for $x>\frac 12$ and $w(x)=\frac 12$ for $x\leq \frac 12$. In an optimal solution, every bin may have at most two items, one of which may be larger than $0.5$, and therefore the total weight of the bin is at most $1.5$. For every batch, consider a solution where items larger than $\frac 12$ are packed into separate bins and items of sizes at most $\frac 12$ are packed in pairs (possibly excluding one bin if their number is odd). The total weight is at least the number of bins minus $\frac q2$.
Note that the additive term has to be linear in $q$. Consider an input with $q$ items of size $\frac 12$. An optimal solution for the entire input has $\lfloor \frac q2 \rfloor$ bins with two items of size $\frac 12$ each, and one bin with a single item if $q$ is odd. If every batch has one item, the number of bins for the batched solution is $q$, and the additive term is at least $\frac {q-3}4$.

The case of $k=3$ and two batches will be covered by the general proof for $q=2$ and $k\geq 3$, and we consider the case of three batches and $k=3$, since for $q\geq 3$, the general case is for $k\geq 4$. In this case we can also prove the upper bound for $q\geq 3$ batches. The function $w$ is defined as follows. Let $w(x)=1$ for $x>\frac 12$, $w(x)=\frac 12$ for $\frac 13 < x \leq \frac 12$, and $w(x)=\frac 13$ for $x\leq \frac 13$. A bin of an optimal solution has weight of at most $\frac{11}6$. By packing the three intervals of items separately for every batch, the total weight is at least the number of bins minus $\frac{7q}6$. The upper bound for other cases has some resemblance to that of \cite{Epstein16}, but the cardinality bounds are taken into account, so additional features are used.

Next, consider the case $k\geq 3$ for two batches. An optimal solution for the entire input is considered, and it is modified into a solution were items of the two batches are packed into separate bins. Such a solution cannot be better than an optimal batched solution. In order to do this, we split bins of an optimal solution $OPT$ into types. A new packing that does not combine items of different batches into one bin is created.
Every type is based on the total size and the number of items that belong to each batch, and we consider bins of one type together and modify the packing. For a bin of $OPT$ that has no items of some batch, there is no need to modify the packing, and it still requires a single bin. Letting $Z$ be the number of such bins, the contributions to the cost of $OPT$ and to the modified packing are both $Z$.
All other bins of $OPT$ have items of both batches and we split every such bin into two bins according to the batches. In several cases we apply repacking in order to reduce the number of bins.

Consider the bins of $OPT$ (not included in $Z$) for which both the total size of items of the first batch is at most $\frac 12$ and the number of items is at most $\frac k2$. Let $Y_1$ be the number of such bins. For the items of the second batch of these bins, the bins for the second batch are unchanged. As for the first batch, pairs of bins are combined, which is always possible. The resulting number of bins is $Y_1+\lceil \frac{Y_1}2 \rceil$. Similarly, for bins of $OPT$ where the total sizes are numbers of the second batch are at most $\frac 12$ and at most $\frac k2$, respectively, and were not included in $Y_1$ (or $Z$), we let their number be $Y_2$, and by repacking bins of the second batch there will be at most $Y_2+\lceil \frac{Y_2}2 \rceil$ bins. Note that for every bin of $OPT$, this bin is valid, so it is not possible that both for the items of the first batch and for those of the total size will be above $\frac 12$. Similarly, it cannot be the case that for both the items of the first batch and the items of the second batch, there are more than $\frac k2$ items. We are left with just one case, where for one of the two batches the total size is above $\frac 12$ and the number of items is at most $\frac k2$, and for the other batch the total size is at most $\frac 12$ and the number of items is above $\frac k2$. We will denote the numbers of bins of these types by $X_1$ and $X_2$. The repacking is similar (with the roles of batches reversed), and we consider the case where the bins for one batch are such that the total size of items is at most $\frac 12$, but the number of items is above $\frac k2$. Since every such bin has at least one item in the other batch, the number of items is at most $k-1$. There is no repacking for the other batch.
The repacking process for this batch is as follows. For $\lfloor \frac{X_1}k\rfloor$ bins, the bins are destroyed, and each item is added to a different bin of the batch. Every bin can receive one additional item, and since its previous total size is at most $\frac 12$, and so was the total size for the original bin of the item, the packing is valid. The number of remaining bins is $X_1-\lfloor \frac{X_1}k\rfloor$, which is at least $\frac{k-1}k \cdot X_1$ and at most $\frac{k-1}k \cdot X_1+1$. The number of repacked items is at most $(k-1)\cdot \lfloor \frac{X_1}k\rfloor \leq \frac{k-1}k \cdot X_1$, so all repacked items are indeed packed into a bin.

Thus, an optimal solution had $Z+Y_1+Y_2+X_1+X_2$ bin, while the solution for the two batches has at most $$Z+\frac 32 (Y_1+Y_2) +1 +(1+\frac{k-1}k)X_1+X_2+2$$ bins. Since $2-\frac 1k \geq \frac 32$, the ratio between the two costs (neglecting the additive term of $3$) is at most $2-\frac 1k$.

For the case with three batches, we assume $k\geq 4$. We apply a similar procedure to the case of two batches. For bins of OPT where there does not exist an item of some batch, only two bins are created. For a bin where there are two batches with total size at most $\frac 12$ and a number of items at most $\frac k2$, neglecting an additive constant, there are two bins for some batches that will combined with other bins of these batches, and on average there are two bins for the bin of OPT. We are left with the case that for one batch the total size is above $\frac 12$ and the number of items is at most $\frac k2$, for one batch both are at most $\frac 12$ and $\frac k2$, respectively, and there is one batch for which the total size is at most $\frac 12$, but the number of items is above $\frac k2$. Since there is at least one item for every batch, there are at most $k-2$ items for this batch.  The reorganization for this batch is similar to the case of two batches, but the difference is that some bins will receive one item from one bin another item from another bin (most bins receive two items coming from one bin, but if $k$ is odd, $k-2$ is odd too, and a destroyed bin cannot be split into pairs of items only).
For that, the single item will be the smallest in its previous bin, and its size will be at most $\frac 12 \cdot \frac 1{k-2} \leq \frac 14$, so adding two such items to a bin whose total size is at most $\frac 12$ results in a valid bin.
\end{proof}

This is the generalization of CCBP, since CCBP can be seen as vector packing in $d\geq 2$ dimensions, where the first component is the size and the other components are equal to $\frac 1k$. Thus, we already found a lower bound of $4$ on the PoC for $d=2$. Here, if FF is applied, an item can be packed into a bin if the resulting set of items will have a sum not larger than $1$ in every component.

\begin{lemma}
For any input $I$ for VP with the parameter $d \geq 2$, it holds that $\sum_{j=1}^{\ell} OPT_j \leq 2d\cdot OPT$.
\end{lemma}

\noindent\begin{proof}Consider a bin of an optimal solution. The total size of all components of all items of any bin is not larger than $d$, since the total size for every component is not larger than $1$.
Let the total size of all components of all items be $X$. We have $OPT\geq \frac{X}{d}$.

Consider a cluster for $I_j$, and an solution of FF for it, which contains $A_j\geq 2$ bins. Consider two arbitrary bins. The first item of the second bin of these bins could not be packed into the first one. Already at the time of arrival of this item, there exists a component for which the sum of the contents of the two bins is above $1$. Thus, for every pair of bins, the total size of all components for the two bins together is above $1$.  Since there are $A_j \choose 2$ pairs, and every bin participates in $A_j-1$ pairs, we get that the sum of all components of all items is at least $\frac{A_j(A_j-1)}2 / (A_j-1)=\frac{A_j}2$.

We have $\sum_{j=1}^{\ell} OPT_j \leq \sum_{j=1}^{\ell} A_j \leq \sum_{j=1}^{\ell} 2\cdot X_j = 2\cdot X \leq 2d \cdot OPT$.
\end{proof}

\begin{lemma}
The PoC of VP with the parameter $d \geq 1$ is at least $2\cdot d$, even if all the components are not larger than a parameter $0<\beta\leq 1$.
\end{lemma}
\begin{proof}
Let $N>2$ be a large positive integer.

The input consists of $d\cdot N(N+1)$ items, partitioned into $d$ types. An item of type $i$ (for $i=1,2,\ldots,d$) has one non-zero component, which is its $i$th component, and its value is $\frac 1N$.
An optimal solution for the entire input has $N+1$ bins in total, each with $N$ items of every type.

There are $d\cdot N$ clusters, where every cluster has just one type of items, and there are $N$ clusters with $N+1$ items of type $i$ for every $i$. Every cluster is packed into two bins by an optimal solution for it. Thus, the cost for the clusters is $2Nd$.
The PoC is at least $\frac{2Nd}{N+1}$, which tends to  $2d$ as $N$ grows to infinity.
\end{proof}

A corollary of this result is that the PoC for the parametric case of one-dimensional classic bin packing remains $2$ for any parameter $\beta$.

\section{Greedy algorithms}\label{gree}

We now analyze WF and NF. The interesting feature is that unlike standard bin packing, the performance of these two algorithms is different here.

\begin{theorem}
The asymptotic approximation ratio of WF for CCBP is $3-\frac{3}{k}$, and for NF it is $3-\frac 2k$, for $k\geq 2$.
\end{theorem}
\begin{proof}
We start with the upper bounds. For NF, let $w(x)=\frac{2(k-1)}k\cdot x+\frac 1k$. A bin of an optimal solution has at most $k$ item with total size at most $1$, and the total weight is at most $\frac{2(k-1)}k+k\cdot \frac 1k =\frac{3k-2}{k}$.
For NF, let $X$ be the number of bins with $k$ items, and let $Y$ be the number of bins whose number of items is at least $1$ and at most $k-1$. The total number of bins is $X+Y$, and we show that the total weight is at least $X+Y-1$.
For every bin with at most $k-1$ items that is not the last bin, the sum of its load and the load of the following bin is above $1$. The sum of all these values (total loads of pairs of consecutive bins) is above $Y-1$ for all such bins, but it is possible that some bins were considered both as a bin with at most $k-1$ items and as a bin following such a bin (if two consecutive bins have at most $k-1$ items, and the second one is not the last bin), and therefore, since the total size for every bin was considered at most twice, the total size of items is above $\frac{Y-1}2$. The number of items is at least $k\cdot X +Y$, and thus the total weight is at least $\frac{k-1}k \cdot (Y-1) + \frac{kX+Y}k > X+Y-1$.

For WF, let $w(x)=\frac{2(k-2)}k\cdot x+\frac 1k$ if $x\leq \frac 12$ and otherwise $w(x)=\frac{2(k-2)}k\cdot x+\frac 2k$. That is, an item of size above $\frac 12$ has an extra addition of $\frac 1k$ to its weight. In this case the weight of at most $k$ items of total size at most $1$ does not exceed $\frac{2(k-2)}k+\frac{k+1}k=\frac{3k-3}{k}$, since at most one item has size above $\frac 12$. For NF, we use similar notation, and let $X$, $Y$, and $Z$ denote the numbers of bins with $k$ items, bins with at least two items and at most $k-1$ items, and bins with a single item, respectively. The number of items is at least $k\cdot X+2\cdot Y+Z$. Due to the action of WF, the number of items of sizes above $\frac 12$ is at least $Z-1$, since no two items packed alone into bins by WF can be packed together, while every two items of sizes at most $\frac 12$ could have been packed together in terms of their sizes (without additional items).  The argument regarding the total size is as before, and given the modified notation, the total size of items is above $\frac{Y+Z-1}2$. Thus, the total weight is at least $$\frac{Z-1}k+\frac{kX+2Y+Z}k+\frac{k-2}k\cdot (Y+Z-1) > X+Y+Z -1 \ . $$

For the lower bound of NF, let $N>0$ be a large integer, which is divisible by $k$, and let $\eps>0$ be a small value such that $\eps <\frac{1}{10N}$. There are $N-1$ large items of sizes in $(\frac 12,\frac 12+N\eps]$ (this interval is contained in $(0.5,0.6)$), $N-2$ medium items of sizes in $(\frac 12 -N\eps, \frac 12)$ (this interval is contained in $(0.4,0.5)$), and there are also $N(k-2)$ small items, each of size $\frac{\eps}k$.
The large items have sizes of $\frac 12+q\cdot \eps$ for $2 \leq q \leq N$. The medium items have sizes of $\frac 12-q\cdot \eps$ for $1 \leq q \leq N-2$.

A possible solution has $N$ bins, where every bin has $k-2$ small items, and one or two other items, where the total size of other items is at most $1-\eps$, which are defined below. Out of these bins, there are $N-4$ bins, where the $i$th bin has one item of size $\frac 12+(i+1)\cdot \eps$ and one item of size $\frac 12 -(i+2)\eps$. The total size for these two items is $1-\eps$, and this leaves five unpacked items that are large or medium, with sizes $\frac 12+(N-2)\eps$, $\frac 12+(N-1)\eps$, $\frac 12+N\eps$,
$\frac 12-2\eps$, and $\frac 12-\eps$. The first three items are packed into separate bins, and the two last items are packed together.

NF receives the items in the following order. First, all small items arrive, and they are packed into $\frac{N(k-2)}{k}$ bins that receive $k$ items each and will not receive other items.
Large and medium items arrive such that the input alternates between them, starting and ending this part of the input with large items. The large items are sorted by decreasing size and the medium items are sorted by increasing size. We get that for this part of the input (after the small items) for $i=1,2,\ldots,N-2$, items of indices $2i-1$ and $2i$ have sizes of $\frac 12 + (N+1-i)\eps$ and $\frac 12 - (N-1-i)\eps$, respectively, and the item of index $2N-3$ has size $\frac 12+2\eps$ (which corresponds to the case $i=N-1$, but there is no item of index $2N-2$). The first item is packed into a new bin since the previous bin has $k$ items. We show that every pair of items have a total size above $1$, and thus every medium or large item is packed into its own bin. If this holds, the number of bins is  $\frac{N(k-2)}{k}+2N-3=N(3-\frac 3N-\frac 2k)$, which implies the lower bound by letting $N$ grow without bound.
An item in an even index in the second part of the input has size $\frac 12 - (N-1-i)\eps$, where the item before it has size  $\frac 12 + (N+1-i)\eps$ and the item after it has size  $\frac 12 + (N-i)\eps$. The total sizes are $1+2\eps$ and $1+\eps$, respectively. Note that this construction is valid for $k=2$, but in that case there are no small items.

The lower bound for $k=2$ and WF is proved separately. Assume that $k\geq 3$.
For a lower bound of WF, let $N>0$ be a large integer, which is divisible by $k$, let $\eps>0$ be a small value such that $\eps <\frac{1}{10N}$, and let $\delta=\frac{\eps}{2^{N+4}}$.
The input consists of $N-1$ huge items, each of size $\frac 12+\delta$,  $N(k-3)$ small items of size $\frac{\delta}k$, $N$ large items, where the $i$th item (for $i=1,2,\ldots,N-4$) has size $\frac 12-\frac{\eps}{2^i}$, and $N$ medium items, where the $i$th medium item has size $\frac 53 \cdot \frac{\eps}{2^i}$. A possible solution has $N$ bins with at most three items that are not small packed into every bin, where their total size is at most $1-\delta$, and the bin also has $k-3$ small items. For $i=1,2,\ldots,N-1$, there is a bin with one huge item, one large item of size $\frac 12-\frac{\eps}{2^{i}}$ and one medium item of size $\frac 53 \cdot \frac{\eps}{2^{i+1}}$, where the total size of these three items is $$(\frac 12 +\delta)+(\frac 12-\frac{\eps}{2^{i}})+(\frac 53 \cdot \frac{\eps}{2^{i+1}})=1+\delta-\frac{\eps}{6\cdot 2^i} \leq 1+\delta-\frac{\eps}{6\cdot 2^N} <1-\delta \ , $$ since $\frac{\eps}{6\cdot 2^N} > \frac{\eps}{2^{N+3}}=2\delta$. The two remaining items that are not small have sizes of $\frac 12-\frac{\eps}{2^{N}}$ and $ \frac{5\eps}{6}$, and they are packed together (with $k-3$ small items).

For WF, the small items are presented first, and they are packed into $\frac{N(k-3)}k$ bins that cannot be reused. Then, pairs of a large item of size $\frac 12-\frac{\eps}{2^{i}}$ and one medium item of size $\frac 53 \cdot \frac{\eps}{2^{i}}$, are presented for $i=1,2,\ldots,N$. The total size of such a pair of items is $\frac 12-\frac{\eps}{2^{i}}$ and one medium item of size $\frac 53 \cdot \frac{\eps}{2^{i}}$ is $\frac 12+  \frac{\eps}{3\cdot 2^{i-1}}$. The first pair is packed into a bin since all previous bins already have $k$ items each. Every time that an item of size $\frac 12-\frac{\eps}{2^{i+1}}$ is presented, it cannot be packed into previous bins, since the minimum load of bins with less than $k$ items is  $\frac 12+  \frac{\eps}{3\cdot 2^{i-1}}$, and $(\frac 12-\frac{\eps}{2^{i+1}})+(\frac 12+  \frac{\eps}{3\cdot 2^{i-1}})=(\frac 12-\frac{\eps}{4\cdot 2^{i-1}})+(\frac 12+  \frac{\eps}{3\cdot 2^{i-1}})>1$. When the medium item of size $\frac 53 \cdot \frac{\eps}{2^{i+1}}$ arrives, the last bin has load below $\frac 12$ while other bins that can receive it have loads above $\frac 12$, and it is combined into the last bin, creating a load of $\frac 12+  \frac{\eps}{3\cdot 2^{i}}$. Finally, all huge items are presented, and each one is packed into a new bin because previous bins either have $k$ items or have loads above $\frac 12$.
Thus, WF has $\frac{N(k-3)}k+N+(N-1)$ bins, and the lower bound on the approximation ratio is $\frac{3k-3}k$ by letting $N$ grow without bound.

The lower bound for the case $k=2$ for WF follows from very simple inputs with $2N$ items of size $0.4$ followed by $2N$ items of size $0.6$. WF creates pairs of items of size $0.4$ and the other items are packed into separate bins, while an optimal solution has $N$ bins with one item of each size.
\end{proof}

\vspace{0.1cm}

Next, we perform an analysis for several closely related special cases of $\beta$, for NF and WF. The goal is to show the difficulty of solving the general case for different values of $\beta$.
We start with the case $\beta = 0.4$.
\begin{theorem}
The asymptotic approximation ratio for WF and NF for CCBP, $k\geq 4$, and   $\beta = 0.4$ is $\frac 83 - \frac{10}{3k}$.
\end{theorem}
\begin{proof}
For the upper bound, we define the following weight function $w(x)=\frac 1k+ \frac{5(k-2)}{3k} \cdot x$. For any bin of an optimal solution, since there are at most $k$ items whose total size is at most $1$, we get a total weight of at most $k\cdot \frac 1k + \frac{5(k-2)}{3k} =\frac 83 -\frac{10}{3k}$.

Consider the output of WF or NF. Since no item has size above $\beta$, every bin that is not the last bin and has at most $k-1$ items has total size above $1-\beta= 0.6$. Since $\beta<0.5$, every such bin has at least two items. The weight of every bin with $k$ items is at least $k\cdot \frac 1k=1$. For other bins, the weight of every bin except for possibly the last one is at least $2\cdot \frac 1k+ 0.6 \cdot \frac{5(k-2)}{3k} =1$.

To prove a lower bound, we define a set of bins which act as an offline solution, and additionally define an ordering of the items for NF and WF. Let $N,M$ be large positive integers, where $N$ is divisible by $k$, and let $\eps>0$ be a small value such that $3^{2M}\cdot\eps < \frac 1{100}$. Every bin has $k-4$ or $k-3$ items of size zero, depending on the number of other items. The larger items defined next.

The offline packing is as follows. There are $3^0\cdot N=N$ bins containing three items of size $0.2+3^{2M}\cdot \eps$, whose total size is below $0.63$. For $i=1,2,\ldots, M-1$, there are $3^i\cdot N$ bins containing three items of size $0.2+3^{2(M-i)}\cdot \eps< 0.21$ and one item of size $0.4-3^{2M-2i+1}\cdot \eps > 0.39$, where the total size of these items is $3(0.2+3^{2(M-i)}\cdot \eps)+(0.4-3^{2M-2i+1}\cdot \eps)=1$. Finally, there are $3^M\cdot N$ bins with one item of size $0.2+\eps<0.21$, one item of size $0.4-\frac{\eps}3$, and one item of size $0.4-3\eps$, whose total size is below $1$.
The total number of bins in this solution is $\Delta=N\cdot \sum_{i=0}^{M} 3^i = N\cdot 3^M \cdot \sum_{i=0}^M 3^{-i}<1.5\cdot N\cdot 3^M $, using the sum of an infinite geometric series.
The bins that have $k-3$ items of size zero are the first $N$ bins and the $3^M$ bins with two items of sizes close to $0.4$, and all other bins have $k-4$ items of size zero. Thus, the number of items of size zero is $(k-3)\cdot \Delta+ M\cdot(3^0+3^M)$.

We define the order in which NF and WF receive the items. First, all items of size zero arrive, every bin will have $k$ such items, and the number of such bins is $\frac{\Delta(k-4)+N\cdot (3^M+1)}k$, which is an integer since $N$ is divisible by $k$ and $\Delta$ is divisible by $N$.
Every item of slightly smaller than $\beta=0.4$ is followed by an item of size slightly smaller than $0.2$.
Specifically, the items of sizes approximately $\beta$ are presented ordered by non-decreasing size, and every item of size $0.4-3^{2j-1}\cdot \eps$ (the number of such items is $N\cdot3^{M-j+1}$ for $1\leq j \leq M$ and $N\cdot3^M$ for $j=0$) is followed by an item of size $0.2+3^{2j}\cdot \eps$, for $j=M,M-1,\ldots,1,0$ (the number of such items is also $N\cdot3^{M-j+1}$ for $1\leq j \leq M$ and $N\cdot3^M$ for $j=0$). We claim that both algorithms pack one item of size approximately $\beta$ and one item of size approximately $0.2$ (two items that arrive consecutively) into every bin. The number of additional bins is therefore $$N\cdot (3^M+ \sum_{j=1}^M 3^{M-j+1})=N\cdot 3^M\cdot(1+\sum_{j=1}^M 3^{1-j})=N\cdot 3^M\cdot(1-\frac 1{3^M}+\sum_{j=0}^M 3^{-j})=\Delta+N\cdot 3^M(1-\frac{1}{3^M}) \ . $$

The ratio between the number of bins for WF and NF and $\Delta$ is $$\frac{k-4}k+\frac{N(3^M+1)}{k\Delta}+ 1+\frac{N\cdot 3^M(1-\frac{1}{3^M})}{\Delta} \ . $$ By using $\Delta <1.5\cdot N\cdot 3^M$, we get a ratio of at least
$2-\frac 4k+\frac{2(1+\frac 1{3^M}  )}{3k}+\frac{2(1-\frac{1}{3^M})}{3}$. Letting $M$ grow to infinity, we get a lower bound of $2-\frac 4k+\frac{2}{3k}+\frac{2}{3}=\frac 83 -\frac{10}{3k}$. It is left to justify the calculation in the sense that the packing of NF and WF is as described. The bins with zero size items cannot receive any other items. We total sizes of pairs of items with index $j$ for which we claim that they are packed together are $0.6+2\cdot 3^{2j-1}\cdot \eps$. All further items of sizes close to $0.4$ that arrive later have sizes of at least $0.4-3^{2j-1}\cdot \eps$, and are too large to be packed into earlier bins. For items of sizes close to $0.2$, NF will pack every such item into the active bin and not into an earlier bin. For WF, all earlier bins will have loads above $0.6$ while the currently last bin has a smaller load, and the item of size close to $0.2$ will be packed there.
\end{proof}

To illustrate the difficulty of analysis for parametric cases further, we focus on WF. We will prove that the same bound can be proved for WF for the case where $\beta \in (0.4,\frac 5{12}\approx 0.41666]$  and $4 \leq k \leq 5$, but for $k\geq 6$, we now show that the bound is higher for example for $\beta=0.41$.
Consider an input with $N$ bins (for a large integer $N$ divisible by $k$), where there are $(13k-45)N$ items of size zero, $11N$ items of size $0.41$, $11N$ items of size $0.4099$, $11N$ items of size $0.1801$, $10N$ items of size $0.1804$, and $2N$ items of size $0.0902$. An optimal solution has $13N$ bins, where the first $11N$ bins have (each) one item of each size out of $0.41$, $0.4099$, and $0.1801$, and $k-3$ items of size zero, and the other $2N$ bins have (each) five items of size $0.1804$, one item of size $0.0902$, and $k-6$ items of size zero. The input starts with items of size zero, and both algorithms have $\frac{(13k-45)N}{k}$ bins with $k$ items of size zero. Then items of size $0.4099$ alternate with items of size $0.1804$ or pairs of items of sizes $0.0902$, and then items of size $0.41$ alternate with items of size $0.1801$. In this process, similar to the construction for $\beta=0.4$, an additional set of $22N$ bins is created. Thus, the ratio is at least $\frac{35-45/k}{13}=\frac{35}{13}-\frac{45}{13k}$. This value is indeed strictly larger than $\frac 83-\frac{10}{3k}$ for $k\geq 6$ (for example, in the case $k=6$ we get above $2.115$ rather than approximately $2.111$.

\begin{theorem}
The asymptotic  approximation ratio for WF for CCBP and   $\beta \in (0.4,\frac 5{12}\approx 0.41666]$ is $\frac {11}6$ for $k=4$ and $2$ for $k=5$.
\end{theorem}
\begin{proof}
Let $w(x)=\begin{cases} 
\frac 1k+ \frac{2(k-2)}{3k}  \  \ \ \ \ { \mbox for \  \ \ }  x  \in (\frac 13,\beta]\\
\frac 1k+ \frac{k-2}{3k}  \ \ \ \ \ \ \ \ \ {\mbox for \  \ \ }  x  \in (\frac 16,\frac 13] \ \ \\
\frac 1k   \ \ \ \ \ \ \ \ \  \ \ \ \ \ \ \ \ \ \ \ {\mbox for  \  \  \ }  x  \in (0,\frac 16]\\
\end{cases}$ be a weight function, which we will use in the analysis.


Any bin of an optimal solution has at most $k$ items, where there are at most five items whose sizes are above $\frac 16$. Moreover, every item of size above $\frac 13$ can be counted as two items of sizes above $\frac 16$, and therefore the total weight is at most $k\cdot \frac 1k + 5\cdot \frac{k-2}{3k} = \frac 83 - \frac{10}{3k}$, which is equal to the claimed bounds for $k=4,5$.

It is left to show that almost every bin of WF has weight of at least $1$. This holds for bins with $k$ items. All other bins that are not the last bin have total sizes above $1-\beta \geq \frac 7{12}$, and thus every bin has at least two items, since $\beta \leq \frac{5}{12}$.
We claim that except for at most two such bins, every bin with at most $k-1$ items either has a total size of items above $\frac 23$, or the bin has at least one item of size above $\frac 13$ (or both). Note that if a bin has total size above $\frac 23$ and it has exactly two items, at least one of them has size above $\frac 13$, and if it has at least four items, at least one of them has size above $\frac 16$.
For that, consider the first bin with total size at most $\frac 23$ and at most $k-1$ items. If this is the last bin or there is no such bin, we are done. Otherwise, for every bin of WF opened later, the first item has size above $\frac 13$.

Consider a bin of WF that has at least two and at most $k-1$ items, and it is not the last one, and it is not a bin with total size at most $\frac 23$ and no item of size above $\frac 13$. Assume that the bin has exactly two items. If the bin has total size above $\frac 23$, then one of the items has size above $\frac 13$. Otherwise, the bin also has such an item. The size of the second item is above $(1-\beta)-\beta \geq \frac 16$. Thus, the total weight is at least  $2\cdot \frac 1k+\frac{3(k-2)}{3k}=1$. If the bin has four items, since $4 \leq k-1$, this situation occurs only for $k=5$, and we show that at least one item has size above $\frac 16$. Indeed, it either has an item with size above $\frac 13$, or otherwise the total size of $k-1 = 4$ items above $\frac 23$, in which case the largest one has size above $\frac 16$. For $k=5$ and a bin with four items, the total weight is $(k-1)\cdot \frac 1k+\frac{k-2}{3k}=\frac{3k-3+k-2}{3k} = 1$.  In the case of a bin with three items, we claim that out of the three items, there is at least one of size above $\frac 13$ or else there are two items of sizes above $\frac 16$. Indeed, if the total size is above $\frac 23$, the largest item has size no larger than $\frac 13$ and the two other items have sizes at most $\frac 16$, we reach a contradiction. Thus, for three items, the total weight is at least $3\cdot \frac 1k+\frac{2(k-2)}{3k}=\frac{9+2k-4}{3k} \geq1$ for $k=4,5$.
\end{proof}

%
For FF, we are able to find tight bounds for any $\beta$.
Recall that FF was already fully analyzed for CCBP with all values of $k$ \cite{KSS75,DosaE18}. The lower bounds hold for any $\beta \in (\frac 12,1]$, and therefore we consider the case $\beta \leq \frac 12$.
Assume that $\beta\in (\frac 1{t+1},\frac 1t]$ for an integer $t\geq 2$.
The case where $k\leq t$ is trivial since $k \cdot \beta \leq \frac kt \leq 1$, which means that FF packs $k$ items into every bin and the packing is optimal.

\begin{theorem}
The asymptotic approximation ratio of FF for CCBP and $\beta\in (\frac 1{t+1},\frac 1t]$ for an integer $t\geq 2$ is $1+\frac{(k-t)(t+1)}{kt}$ for $k \geq t$.
\end{theorem}
\begin{proof}
The case $k=t$ was discussed above and the approximation ratio is indeed $1$. In what follows we assume that $k>t$ holds.
For simplicity we allow zero size items in this lower bound. There is a known lower bound construction for any $\beta \in (\frac 1{t+1},\frac 1t]$, such that every bin of an optimal solution has $t+1$ items while almost all bins of FF have $t$ items. Letting $N$ be the number of bins of an optimal solution such that $N$ is divisible by $k$, we add $k-t-1$ items of size zero to each such bin. If the zero size items arrive first, FF has $\frac{N(k-t-1)}{k}$ bins for them. It also has approximately (up to an additive constant number of bins) $\frac{N\cdot (t+1)}t$ bins for the other items. By letting $N$ grow without  bound, The approximation ratio is at least $$\frac{k-t-1}{k}+\frac{t+1}t=\frac{t(k-t-1)+k(t+1)}{kt}=1+\frac{(k-t)(t+1)}{kt} \ . $$

For an upper bound, we use the weight function $w(x)=\frac{k-t}{k}\cdot \frac{t+1}{t}\cdot x +\frac 1k$.  For a bin of an optimal solution with total size at most $1$ and at most $k$ items, the total weight is at most
$w(x)=\frac{k-t}{k}\cdot \frac{t+1}{t}+ k\cdot \frac 1k =1+\frac{(k-t)(t+1)}{kt}$.

We use a simple analysis of FF. For bins with $k$ items the total weight is at least $1$, since the weight of every item is at least $\frac 1k$, and we consider other bins. For other bins, FF without cardinality constraints was applied, and every bin except for possibly the last bin and at most one additional bin has items of total size at least $\frac {t}{t+1}$. This holds since once a bin has a smaller load (and at most $k-1$ items), every item packed into a later bin has size above $\frac {1}{t+1}$. Since no item has size above $\frac 1t$, every further bin (except for possibly the last bin) has exactly $t$ items of size at least $\frac{t}{t+1}$. Moreover, every bin has at least $t$ items, since $t<k$ and every item has size at most $\frac 1t$. Thus, except for possibly two bins, every bin with at most $k-1$ items has total weight of at least $\frac{k-t}{k}\cdot \frac{t+1}{t}\cdot \frac{t}{t+1} +t\cdot \frac 1k=1$.
\end{proof}



\bibliographystyle{abbrv}

\bibliography{pocv}

\begin{thebibliography}{10}

\bibitem{ACFR16}
Y.~Azar, I.~R. Cohen, A.~Fiat, and A.~Roytman.
\newblock Packing small vectors.
\newblock In {\em Proc. of the 27th Annual {ACM-SIAM} Symposium on Discrete
  Algorithms, (SODA2016)}, pages 1511--1525, 2016.

\bibitem{AESV}
Y.~Azar, Y.~Emek, R.~van Stee, and D.~Vainstein.
\newblock The price of clustering in bin-packing with applications to
  bin-packing with delays.
\newblock In {\em The 31st {ACM} on Symposium on Parallelism in Algorithms and
  Architectures, (SPAA2019)}, pages 1--10, 2019.

\bibitem{BCKK04}
L.~Babel, B.~Chen, H.~Kellerer, and V.~Kotov.
\newblock Algorithms for on-line bin-packing problems with cardinality
  constraints.
\newblock {\em Discrete Applied Mathematics}, 143(1-3):238--251, 2004.

\bibitem{BC81}
B.~S. Baker and E.~G. {Coffman, Jr.}
\newblock A tight asymptotic bound for next-fit-decreasing bin-packing.
\newblock {\em SIAM J. on Algebraic and Discrete Methods}, 2(2):147--152, 1981.

\bibitem{BBDEL_ESA18}
J.~Balogh, J.~B{\'{e}}k{\'{e}}si, G.~D{\'{o}}sa, L.~Epstein, and A.~Levin.
\newblock A new and improved algorithm for online bin packing.
\newblock In {\em Proc. of the 26th European Symposium on Algorithms
  (ESA2018)}, pages 5:1--5:14, 2018.

\bibitem{BBDEL_ESA}
J.~Balogh, J.~B{\'{e}}k{\'{e}}si, G.~D{\'{o}}sa, L.~Epstein, and A.~Levin.
\newblock Online bin packing with cardinality constraints resolved.
\newblock {\em Journal of Computer and System Sciences}, 112:34--49, 2020.

\bibitem{BBDEL_newlb}
J.~Balogh, J.~B{\'{e}}k{\'{e}}si, G.~D{\'{o}}sa, L.~Epstein, and A.~Levin.
\newblock A new lower bound for classic online bin packing.
\newblock {\em Algorithmica}, 83(7):2047--2062, 2021.

\bibitem{BBDGT}
J.~Balogh, J.~B\'{e}k\'{e}si, G.~D\'{o}sa, G.~Galambos, and Z.~Tan.
\newblock Lower bound for 3-batched bin packing.
\newblock {\em Discrete Optimization}, 21:14--24, 2016.

\bibitem{BBDSS19}
J.~Balogh, J.~B{\'{e}}k{\'{e}}si, G.~D{\'{o}}sa, J.~Sgall, and R.~van Stee.
\newblock The optimal absolute ratio for online bin packing.
\newblock {\em Journal of Computer and System Sciences}, 102:1--17, 2019.

\bibitem{BBG}
J.~Balogh, J.~B{\'e}k{\'e}si, and G.~Galambos.
\newblock New lower bounds for certain classes of bin packing algorithms.
\newblock {\em Theoretical Computer Science}, 440:1--13, 2012.

\bibitem{BDE}
J.~B{\'e}k{\'e}si, G.~D\'osa, and L.~Epstein.
\newblock Bounds for online bin packing with cardinality constraints.
\newblock {\em Information and Computation}, 249:190--204, 2016.

\bibitem{CKP03}
A.~Caprara, H.~Kellerer, and U.~Pferschy.
\newblock Approximation schemes for ordered vector packing problems.
\newblock {\em Naval Research Logistics}, 92:58--69, 2003.

\bibitem{Dosa15}
G.~D{\'{o}}sa.
\newblock The tight absolute bound of {First Fit} in the parameterized case.
\newblock {\em Theoretical Computer Science}, 596:149--154, 2015.

\bibitem{Do15}
G.~D\'{o}sa.
\newblock Batched bin packing revisited.
\newblock {\em Journal of Scheduling}, 20(2):199--209, 2017.

\bibitem{DosaE18}
G.~D{\'{o}}sa and L.~Epstein.
\newblock The tight asymptotic approximation ratio of {First Fit} for bin
  packing with cardinality constraints.
\newblock {\em Journal of Computer and System Sciences}, 96:33--49, 2018.

\bibitem{DS12}
G.~D\'osa and J.~Sgall.
\newblock {First Fit} bin packing: A tight analysis.
\newblock In {\em Proc. of the 30th International Symposium on Theoretical
  Aspects of Computer Science (STACS2013)}, pages 538--549, 2013.

\bibitem{DLHT}
G.~Dósa, R.~Li, X.~Han, and Z.~Tuza.
\newblock Tight absolute bound for first fit decreasing bin-packing: ${FFD(L)}
  \leq 11/9 {OPT(L)} + 6/9$.
\newblock {\em Theoretical Computer Science}, 510:13--61, 2013.

\bibitem{Epstein05}
L.~Epstein.
\newblock Online bin packing with cardinality constraints.
\newblock {\em SIAM Journal on Discrete Mathematics}, 20(4):1015--1030, 2006.

\bibitem{Epstein16}
L.~Epstein.
\newblock More on batched bin packing.
\newblock {\em Operations Research Letters}, 44(2):273--277, 2016.

\bibitem{E20}
L.~Epstein.
\newblock On bin packing with clustering and bin packing with delays.
\newblock {\em CoRR}, abs/1908.06727, 2019.
\newblock Also in Discrete Optimization, to appear.

\bibitem{E21}
L.~Epstein.
\newblock Open-end bin packing: new and old analysis approaches.
\newblock {\em CoRR}, abs/2105.05923, 2021.

\bibitem{EL07afptas}
L.~Epstein and A.~Levin.
\newblock {AFPTAS} results for common variants of bin packing: A new method for
  handling the small items.
\newblock {\em {SIAM} Journal on Optimization}, 20(6):3121--3145, 2010.

\bibitem{FerLue81}
W.~{Fernandez de la Vega} and G.~S. Lueker.
\newblock Bin packing can be solved within $1+\varepsilon$ in linear time.
\newblock {\em Combinatorica}, 1(4):349--355, 1981.

\bibitem{FK13}
H.~Fujiwara and K.~M. Kobayashi.
\newblock Improved lower bounds for the online bin packing problem with
  cardinality constraints.
\newblock {\em Journal of Combinatorial Optimization}, 29(1):67--87, 2015.

\bibitem{GJYbatch}
G.~Gutin, T.~Jensen, and A.~Yeo.
\newblock Batched bin packing.
\newblock {\em Discrete Optimization}, 2(1):71--82, 2005.

\bibitem{J74}
D.~S. Johnson.
\newblock Fast algorithms for bin packing.
\newblock {\em Journal of Computer and System Sciences}, 8:272--314, 1974.

\bibitem{JoDUGG74}
D.~S. Johnson, A.~Demers, J.~D. Ullman, M.~R. Garey, and R.~L. Graham.
\newblock Worst-case performance bounds for simple one-dimensional packing
  algorithms.
\newblock {\em SIAM Journal on Computing}, 3:256--278, 1974.

\bibitem{KK82}
N.~Karmarkar and R.~M. Karp.
\newblock An efficient approximation scheme for the one-dimensional bin-packing
  problem.
\newblock In {\em Proceedings of the 23rd Annual Symposium on Foundations of
  Computer Science (FOCS1982)}, pages 312--320, 1982.

\bibitem{KP99}
H.~Kellerer and U.~Pferschy.
\newblock Cardinality constrained bin-packing problems.
\newblock {\em Annals of Operations Research}, 92:335--348, 1999.

\bibitem{KSS75}
K.~L. Krause, V.~Y. Shen, and H.~D. Schwetman.
\newblock {Analysis} of several task-scheduling algorithms for a model of
  multiprogramming computer systems.
\newblock {\em Journal of the ACM}, 22(4):522--550, 1975.

\bibitem{LeeLee85}
C.~C. Lee and D.~T. Lee.
\newblock A simple online bin packing algorithm.
\newblock {\em Journal of the ACM}, 32(3):562--572, 1985.

\bibitem{RaBrLL89}
P.~Ramanan, D.~J. Brown, C.~C. Lee, and D.~T. Lee.
\newblock Online bin packing in linear time.
\newblock {\em Journal of Algorithms}, 10:305--326, 1989.

\bibitem{Vliet92}
A.~van Vliet.
\newblock An improved lower bound for online bin packing algorithms.
\newblock {\em Information Processing Letters}, 43(5):277--284, 1992.

\end{thebibliography}

\end{document}